\documentclass[american,aps,pra,reprint,superscriptaddress,showpacs]{revtex4-1}
\usepackage[T1]{fontenc}
\usepackage[latin9]{inputenc}
\usepackage{babel}
\usepackage{array}
\usepackage{amsthm}
\usepackage{amsmath}
\usepackage{amssymb}
\usepackage{graphicx}
\usepackage[unicode=true,pdfusetitle,
 bookmarks=true,bookmarksnumbered=false,bookmarksopen=false,
 breaklinks=false,pdfborder={0 0 0},backref=false,colorlinks=false]
 {hyperref}
\hypersetup{
 colorlinks,linkcolor=myrefcolor,citecolor=myurlcolor,urlcolor=myurlcolor}

\makeatletter

\providecommand{\tabularnewline}{\\}

\@ifundefined{textcolor}{}
{%
 \definecolor{BLACK}{gray}{0}
 \definecolor{WHITE}{gray}{1}
 \definecolor{RED}{rgb}{1,0,0}
 \definecolor{GREEN}{rgb}{0,1,0}
 \definecolor{BLUE}{rgb}{0,0,1}
 \definecolor{CYAN}{cmyk}{1,0,0,0}
 \definecolor{MAGENTA}{cmyk}{0,1,0,0}
 \definecolor{YELLOW}{cmyk}{0,0,1,0}
 }
\theoremstyle{plain}
\newtheorem{thm}{\protect\theoremname}
\ifx\proof\undefined
\newenvironment{proof}[1][\protect\proofname]{\par
\normalfont\topsep6\p@\@plus6\p@\relax
\trivlist
\itemindent\parindent
\item[\hskip\labelsep
\scshape
#1]\ignorespaces
}{%
\endtrivlist\@endpefalse
}
\providecommand{\proofname}{Proof}
\fi
\theoremstyle{plain}
\newtheorem{conjecture}[thm]{\protect\conjecturename}

\usepackage{babel}
\usepackage{times}
\usepackage{txfonts}
\usepackage{braket}
\usepackage{colortbl}
\definecolor{myurlcolor}{rgb}{0,0,0.7}
\definecolor{myrefcolor}{rgb}{0.8,0,0}
\newcommand{\proj}[1]{\ket{#1}\!\bra{#1}}

\makeatother

\providecommand{\conjecturename}{Conjecture}
\providecommand{\theoremname}{Theorem}

\begin{document}

\title{Progress towards a unified approach to entanglement distribution}

\author{Alexander Streltsov}

\email{alexander.streltsov@icfo.es}

\selectlanguage{american}%

\affiliation{ICFO -- Institut de Ci\`encies Fot\`oniques, Av. C.F. Gauss, 3,  E-08860 Castelldefels, Spain}

\author{Remigiusz Augusiak}

\affiliation{ICFO -- Institut de Ci\`encies Fot\`oniques, Av. C.F. Gauss, 3,  E-08860 Castelldefels, Spain}

\author{Maciej Demianowicz}

\affiliation{ICFO -- Institut de Ci\`encies Fot\`oniques, Av. C.F. Gauss, 3,  E-08860 Castelldefels, Spain}

\author{Maciej Lewenstein}

\affiliation{ICFO -- Institut de Ci\`encies Fot\`oniques, Av. C.F. Gauss, 3,  E-08860 Castelldefels, Spain}

\affiliation{ICREA -- Instituci\'o Catalana de Recerca i Estudis Avan\c{c}ats, Lluis Companys 23, E-08010 Barcelona, Spain}
\begin{abstract}
Entanglement distribution is key to the success of secure communication
schemes based on quantum mechanics, and there is a strong need for
an ultimate architecture able to overcome the limitations of recent
proposals such as those based on entanglement percolation or quantum
repeaters. In this work we provide broad theoretical background for
the development of such technologies. In particular, we investigate
the question of whether entanglement distribution is more efficient
if some amount of entanglement -- or some amount of correlations in
general -- is available prior to the transmission stage of the protocol.
We show that in the presence of noise the answer to this question
strongly depends on the type of noise and on the way how entanglement
is quantified. On the one hand, subadditive entanglement measures
do not show advantage of preshared correlations if entanglement is
established via combinations of single-qubit Pauli channels. On the
other hand, based on the superadditivity conjecture of distillable
entanglement, we provide evidence that this phenomenon occurs for
this measure. These results strongly suggest that sending one half
of some pure entangled state down a noisy channel is the best strategy
for any subadditive entanglement quantifier, thus paving the way to
a unified approach for entanglement distribution which does not depend
on the nature of noise. We also provide general bounds for entanglement
distribution involving quantum discord, and present a counter-intuitive
phenomenon of the advantage of arbitrarily little entangled states
over maximally entangled ones, which may also occur for quantum channels
relevant in experiments.
\end{abstract}


\maketitle

\section{Introduction}

Considered as a curiosity in the early days of quantum theory \cite{Einstein1935},
entanglement has been now recognized as the essential ingredient for
a growing number of applications in quantum technologies \cite{Nielsen2000,Horodecki2009}.
Among them we find for example the celebrated quantum cryptography
\cite{Ekert1991} allowing for a provably secure communication between
distant parties, and quantum teleportation \cite{Bennett1993} which
offers the possibility of an intact transmission of a state of a particle
over an arbitrarily long distance using preshared entanglement and
classical communication. Entanglement is also necessary for quantum
nonlocality, which is an even stronger resource for certain information--processing
tasks, including the above mentioned secure key distribution \cite{Ekert1991,Barrett2005,Acin2007}
and certified quantum randomness generation \cite{Pironio2010,Colbeck2009,Colbeck2011}.

A common assumption behind entanglement-based protocols is that long-distance
or at least medium-distance entanglement is available beforehand.
Several remedies against this drawback have been recently proposed
with the most promising one being based on quantum repeaters \cite{Briegel1998}
and entanglement distillation \cite{Bennett1996a}. However, the necessity
of powerful quantum memories appears as the main limiting factor in
this proposal (cf. \cite{Sangouard2011}). Another method is based
on entanglement percolation \cite{Acin2007a}, but it also suffers
problems when considered in realistic situations in the presence of
noise and decoherence \cite{Perseguers2013}.

The aim of the present work is to explore different realistic scenarios
in which the long-distance entanglement can be distributed. The general
framework for such a task we adopt here is the following (see Fig.
\ref{fig:distribution}). Two parties, Alice and Bob, initially share
a three-particle quantum system. Two of the particles are with Alice,
and the remaining one is in Bob's hands. In the most general situation
we allow Alice and Bob to share some correlations established before
the beginning of the protocol. The distribution of entanglement is
then achieved with the aid of a quantum channel which is used to transmit
one of Alice's particles to Bob. 
\begin{figure}
\includegraphics[width=1\columnwidth]{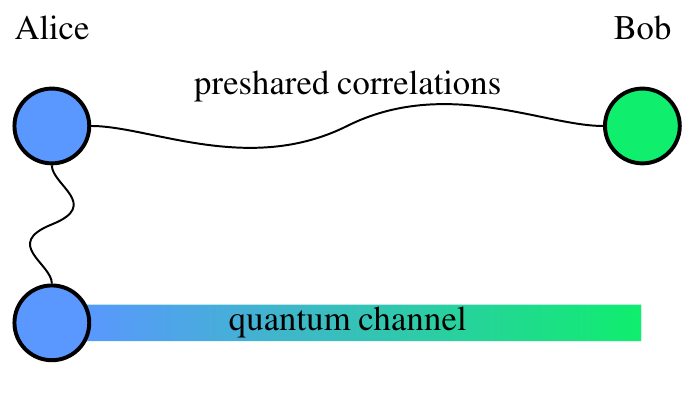} \caption{\label{fig:distribution}General framework for entanglement distribution.
Alice is initially in possession of two particles, while one particle
is in Bob's hands. Alice and Bob further have access to preshared
correlations and an additional -- possibly noisy -- quantum channel
which is used for entanglement distribution.}
\end{figure}

A remarkable result with respect to such general entanglement distribution
protocols has been obtained in \cite{Cubitt2003}. There, it was shown
that the process is even possible without sending entanglement directly:
for successful entanglement distribution the exchanged particle does
not need to be entangled with the rest of the system. This phenomenon
has been termed ``entanglement distribution with separable states'',
and its experimental verification has also been reported recently
\cite{Fedrizzi2013,Vollmer2013,Peuntinger2013,Silberhorn2013}. These
results suggest that such a distribution procedure may be advantageous
in the presence of noise: it could be possible to surpass the fragileness
of entanglement by sending a separable particle.

Despite considerable attempts to understand this phenomenon \cite{Mifmmodeheckslsesita2008,Streltsov2012,Chuan2012,Kay2012,Mifmmodeheckslsesita2013,Streltsov2014},
one of the most important questions remains unresolved: \textit{Can
noisy entanglement distribution with separable states provide an advantage
when compared to sending one half of the maximally entangled state
through the same noisy channel?} Note that the answer for this question
has also a direct importance for the theory and practice of quantum
repeaters and quantum percolation were the intermediate-distance entanglement
between the involved nodes must be established in some way. In this
work we attack this problem by focusing on the following closely related
questions: 
\begin{itemize}
\item Given a noisy quantum channel, what is the maximal amount of entanglement
that can be distributed with and without preshared correlations? 
\item Are preshared correlations helpful for entanglement distribution via
a given quantum channel? 
\end{itemize}
Note that a negative answer to the second question also implies that
entanglement distribution with separable states is not the best strategy
in this situation.

As our study reveals, the answers to these questions depend on the
way entanglement is quantified. In particular, we show that if the
entanglement quantifier is subadditive (that is, its value for a tensor
product of any two states is not greater than the sum of the values
for the individual states) preshared correlations provide no advantage
for single-qubit Pauli channels or any tensor product thereof. In
this situation the best distribution strategy is to send one half
of the maximally entangled state down the noisy channel. However,
not all entanglement quantifiers are subadditive. In particular, it
is conjectured that the distillable entanglement is superadditive
\cite{Shor2001}. Assuming this conjecture holds true, we show that
preshared correlations can indeed provide advantage for the distribution
of distillable entanglement. Another surprising result is obtained
for the logarithmic negativity: for this entanglement measure states
with arbitrarily little entanglement can show better performance for
entanglement distribution when compared to maximally entangled states.
We further present bounds for noisy entanglement distribution given
by quantum discord \cite{Modi2012,Streltsov2014a}, thus significantly
extending the results provided in \cite{Streltsov2012,Chuan2012}
to the noisy scenario.

Moreover, the results presented in this work strongly suggest that
a unified approach to entanglement distribution is indeed possible.
In particular, based on our findings it is very reasonable to assume
that preshared correlations do not provide advantage for any subadditive
entanglement quantifier, regardless of the type of noisy channel used
for the distribution. If this assumption is correct, sending one half
of some pure entangled state down a noisy channel will be the best
strategy in this very general scenario. However, we also show that
maximally entangled states are not necessarily optimal for this process.

This paper is organized as follows: in Section \ref{sec:Noiseless}
we study noiseless entanglement distribution, while the scenario involving
noise is considered in Section \ref{sec:Noisy}. In Section \ref{sec:Optimal}
we investigate optimal entanglement distribution without preshared
correlations, i.e., we consider the maximal amount of entanglement
that can be distributed via a given noisy channel if Alice and Bob
do not share any correlations initially. Finally, the possible advantage
of preshared correlations for noisy entanglement distribution is discussed
in Section \ref{sec:Advantage}.

\section{\label{sec:Noiseless}Noiseless entanglement distribution}

Starting point of this section is the general scenario for entanglement
distribution considered in \cite{Streltsov2012,Chuan2012}, see also
\cite{Streltsov2014a} for a detailed discussion. In particular, we
assume that two parties, Alice and Bob, have access to a general tripartite
quantum state $\rho=\rho^{ABC}$. We further assume -- without loss
of generality -- that the entanglement distribution is realized by
sending the particle $C$ from Alice to Bob, and that during the entire
process the particles $A$ and $B$ are in possession of Alice and
Bob respectively. If the quantum channel used for the transmission
of the particle $C$ is noiseless, the amount of entanglement distributed
in this process is quantified via the difference $E^{A|BC}(\rho)-E^{AC|B}(\rho)$
between the final amount of entanglement $E^{A|BC}(\rho)$ and the
initial amount of entanglement $E^{AC|B}(\rho)$.

As it was shown in \cite{Cubitt2003}, entanglement distribution is
also possible by sending a particle which is not entangled with the
rest of the system, i.e., there exist states $\rho=\rho^{ABC}$ such
that $E^{C|AB}(\rho)=0$ and, at the same time, $E^{A|BC}(\rho)-E^{AC|B}(\rho)>0$.
This finding has triggered a debate about the type of correlations
which are responsible for entanglement distribution. An important
result in this context was provided in \cite{Streltsov2012,Chuan2012}:
the amount of distributed entanglement cannot exceed the amount of
quantum discord $\Delta^{C|AB}$ between the exchange particle $C$
and the remaining system $AB$: 
\begin{equation}
\Delta^{C|AB}(\rho)\geq E^{A|BC}(\rho)-E^{AC|B}(\rho).\label{eq:main}
\end{equation}
At this point, it is also important to notice that in general quantum
discord does not vanish on separable states. This inequality was shown
to hold for all distance-based quantifiers of entanglement and discord
\cite{Streltsov2012}: 
\begin{align}
E^{X|Y}(\rho^{XY}) & =\min_{\sigma^{XY}\in{\cal S}}D(\rho^{XY},\sigma^{XY}),\label{eq:entanglement}\\
\Delta^{X|Y}(\rho^{XY}) & =\min_{\{\Pi_{i}^{X}\}}D(\rho^{XY},\sum_{i}\Pi_{i}^{X}\rho^{XY}\Pi_{i}^{X}).\label{eq:discord}
\end{align}
Here, ${\cal S}$ is the set of bipartite separable states, $\{\Pi_{i}^{X}\}$
is a local von Neumann measurement on the subsystem $X$, and $D$
can be any general distance which satisfies the following two properties
\cite{Streltsov2012}: 
\begin{itemize}
\item $D$ does not increase under quantum operations: 
\begin{equation}
D(\Lambda[\rho],\Lambda[\sigma])\leq D(\rho,\sigma)\label{eq:contractive}
\end{equation}
for any quantum operation $\Lambda$ and any pair of quantum states
$\rho$ and $\sigma$, 
\item $D$ satisfies the triangle inequality: 
\begin{equation}
D(\rho,\sigma)\leq D(\rho,\tau)+D(\tau,\sigma)\label{eq:triangle}
\end{equation}
for any three quantum states $\rho$, $\sigma$ and $\tau$. 
\end{itemize}
As it was further shown in \cite{Streltsov2012,Chuan2012}, the results
presented above also hold for the quantum relative entropy $S(\rho||\sigma)=\mathrm{Tr}[\rho\log\rho]-\mathrm{Tr}[\rho\log\sigma]$,
despite the fact that the relative entropy in general does not satisfy
the triangle inequality. The corresponding quantifiers of entanglement
and discord in this case are known as the relative entropy of entanglement
$E_{R}$ and the relative entropy of discord $\Delta_{R}$: 
\begin{align}
E_{R}^{X|Y}(\rho^{XY}) & =\underset{\sigma^{XY}\in{\cal S}}{\min}S(\rho^{XY}||\sigma^{XY}),\label{eq:REE}\\
\Delta_{R}^{X|Y}(\rho^{XY}) & =\underset{\{\Pi_{i}^{X}\}}{\min}S(\rho^{XY}||\sum_{i}\Pi_{i}^{X}\rho^{XY}\Pi_{i}^{X}).\label{eq:RED}
\end{align}

The relative entropy of entanglement $E_{R}$ was originally introduced
in \cite{Vedral1997,Vedral1998}. By its relation to the relative
entropy \cite{Schumacher,Vedral2002} it plays a fundamental role
in quantum information theory. $E_{R}$ is known to be an upper bound
on the distillable entanglement $E_{d}$ \cite{Rains1999,Horodecki2000}
and a lower bound on the entanglement of formation $E_{f}$ \cite{Vedral1998}:
\begin{equation}
E_{d}\leq E_{R}\leq E_{f}.\label{eq:bound}
\end{equation}
The distillable entanglement $E_{d}$ quantifies the maximal number
of singlets that can be asymptotically obtained per copy of the given
state via local operations and classical communication (LOCC) \cite{Bennett1996a}.
The entanglement of formation $E_{f}$ is defined as \cite{Bennett1996}
\begin{equation}
E_{f}(\rho^{XY})=\min\sum_{i}p_{i}E(\ket{\psi_{i}}^{XY}),
\end{equation}
where the minimum is taken over all pure-state decompositions $\{p_{i},\ket{\psi_{i}}^{XY}\}$
of the state $\rho^{XY}$, i.e., $\rho^{XY}=\sum_{i}p_{i}\ket{\psi}\bra{\psi}^{XY}$,
and $E(\ket{\psi}^{XY})=S(\rho^{X})$ is the von Neumann entropy of
the reduced state. 

The relative entropy of discord $\Delta_{R}$ was originally introduced
in \cite{Horodecki2005}, where it was called ``one-way information
deficit'' %
\footnote{Note that the term ``relative entropy of discord'' first appeared
in \cite{Modi2010}, where it was used for the minimal relative entropy
to the set of fully classical states. Here we will reserve the name
``relative entropy of discord'' exclusively for the quantity given
in Eq. (\ref{eq:RED}).%
}. It quantifies the amount of information which cannot be localized
by one-way classical communication between two parties.

\subsection{Relation to distillable entanglement and entanglement cost}

Equipped with these tools we are now in position to present the first
result of this paper. In particular, we will provide a close connection
between the relative entropy of discord $\Delta_{R}$, the distillable
entanglement $E_{d}$, and the entanglement cost $E_{c}$. The latter
is defined as the minimal number of singlets per copy required for
the asymptotic creation of a bipartite quantum state via LOCC \cite{Horodecki2009},
and can also be written as the regularized entanglement of formation
\cite{Hayden2001}: 
\begin{equation}
E_{c}(\rho)=\underset{n\rightarrow\infty}{\lim}\frac{1}{n}E_{f}(\rho^{\otimes n}).\label{eq:Ec}
\end{equation}
 The aforementioned relation between $\Delta_{R}$, $E_{d}$, and
$E_{c}$ is provided in the following theorem. 
\begin{thm}
\label{thm:1}Given a tripartite state $\rho=\rho^{ABC}$, the following
inequality holds: 
\begin{equation}
\Delta_{R}^{C|AB}(\rho)\geq E_{d}^{A|BC}(\rho)-E_{c}^{AC|B}(\rho).\label{eq:main-1}
\end{equation}
\end{thm}
\begin{proof}
This inequality can be proven by noticing that the inequality (\ref{eq:main})
also holds for the regularized relative entropy of entanglement and
discord: 
\begin{equation}
\underset{n\rightarrow\infty}{\lim}\frac{\Delta_{R}^{C|AB}(\rho^{\otimes n})}{n}\geq\underset{n\rightarrow\infty}{\lim}\frac{E_{R}^{A|BC}(\rho^{\otimes n})}{n}-\underset{n\rightarrow\infty}{\lim}\frac{E_{R}^{AC|B}(\rho^{\otimes n})}{n}.
\end{equation}
By applying Eq. (\ref{eq:bound}) and using the fact that the distillable
entanglement $E_{d}$ does not change under regularization we arrive
at the inequality 
\begin{equation}
\underset{n\rightarrow\infty}{\lim}\frac{\Delta_{R}^{C|AB}(\rho^{\otimes n})}{n}\geq E_{d}^{A|BC}(\rho)-\underset{n\rightarrow\infty}{\lim}\frac{E_{f}^{AC|B}(\rho^{\otimes n})}{n}.
\end{equation}
In the next step we recall that the entanglement cost is equal to
the regularized entanglement of formation, see Eq. (\ref{eq:Ec}),
and thus 
\begin{equation}
\underset{n\rightarrow\infty}{\lim}\frac{\Delta_{R}^{C|AB}(\rho^{\otimes n})}{n}\geq E_{d}^{A|BC}(\rho)-E_{c}^{AC|B}(\rho).
\end{equation}
Finally, the desired inequality (\ref{eq:main-1}) follows by observing
that the relative entropy of discord does not increase under regularization:
$\Delta_{R}(\rho)\geq\underset{n\rightarrow\infty}{\lim}\Delta_{R}(\rho^{\otimes n})/n$. 
\end{proof}
Notice that Eq. (\ref{eq:main-1}) has a clear operational interpretation:
the relative entropy of discord is an upper bound on the number of
singlets gained in the process of entanglement distribution in the
asymptotic limit. This is because $E_{d}^{A|BC}(\rho)$ quantifies
the number of singlets Alice and Bob can distill \textit{after} performing
the entanglement distribution, while $E_{c}^{AC|B}(\rho)$ quantifies
the amount of singlets that Alice and Bob need to create the state
$\rho=\rho^{ABC}$ \textit{before} the entanglement distribution has
been performed, both in the asymptotic limit. Moreover, as already
mentioned in the proof of Theorem \ref{thm:1}, this statement is
also true for the regularized relative entropy of discord $\underset{n\rightarrow\infty}{\lim}\Delta_{R}(\rho^{\otimes n})/n$.

\subsection{Relation to measures of NPT entanglement and distillability}

The results presented above demonstrate that the relation between
entanglement and discord in Eq. (\ref{eq:main}) is more general than
anticipated by the original approach \cite{Streltsov2012,Chuan2012}.
In the following we will go one step further by extending these results
to general measures of NPT (nonpositive partial transpose) entanglement.
In particular, we consider entanglement quantifiers of the form \cite{Horodecki2009}
\begin{equation}
E_{\mathrm{PPT}}^{X|Y}(\rho^{XY})=\min_{\sigma^{XY}\in\mathrm{PPT}}D(\rho^{XY},\sigma^{XY}),
\end{equation}
where $\mathrm{PPT}$ is the set of states having positive partial
transpose, and the distance $D$ satisfies Eqs. (\ref{eq:contractive})
and (\ref{eq:triangle}). The amount of quantum discord is defined
in the same way as in Eq. (\ref{eq:discord}): 
\begin{equation}
\Delta^{X|Y}(\rho^{XY})=\min_{\{\Pi_{i}^{X}\}}D(\rho^{XY},\sum_{i}\Pi_{i}^{X}\rho^{XY}\Pi_{i}^{X}).\label{eq:discord-1}
\end{equation}
The following theorem shows that the inequality (\ref{eq:main}) also
applies to these measures of NPT entanglement. 
\begin{thm}
\label{thm:2}Given a tripartite state $\rho=\rho^{ABC}$, the following
inequality holds: 
\begin{equation}
\Delta^{C|AB}(\rho)\geq E_{\mathrm{PPT}}^{A|BC}(\rho)-E_{\mathrm{PPT}}^{AC|B}(\rho).\label{eq:PPT}
\end{equation}
\end{thm}
\begin{proof}
The proof goes along the lines of the one of Eq. (\ref{eq:main}),
first presented in \cite{Streltsov2012}. We start by introducing
the state $\sigma=\sigma^{ABC}$ which is PPT with respect to the
bipartition $AC|B$, and, moreover, we assume that it is the closest
PPT state to $\rho$: $E_{\mathrm{PPT}}^{AC|B}(\rho)=D(\rho,\sigma)$.
We then define the states 
\begin{equation}
\rho'=\sum_{i}\Pi_{i}^{C}\rho\Pi_{i}^{C}\label{rhoprime}
\end{equation}
and 
\begin{equation}
\sigma'=\sum_{i}\Pi_{i}^{C}\sigma\Pi_{i}^{C}\label{sigmaprime}
\end{equation}
to arise from $\rho$ and $\sigma$ via the local von Neumann measurement
on $C$ minimizing the distance between $\rho$ and $\rho'$, i.e.,
$\Delta^{C|AB}(\rho)=D(\rho,\rho')$. Further, we use the fact that
the distance $D$ satisfies the triangle inequality, and thus 
\begin{equation}
D(\rho,\sigma')\leq D(\rho,\rho')+D(\rho',\sigma').\label{eq:triangle-1}
\end{equation}
Recalling that $D$ does not increase under quantum operations it
follows that 
\begin{equation}
D(\rho',\sigma')\leq D(\rho,\sigma),\label{eq:contractive-1}
\end{equation}
and Eq. (\ref{eq:triangle-1}) becomes 
\begin{equation}
D(\rho,\sigma')\leq\Delta^{C|AB}(\rho)+E_{\mathrm{PPT}}^{AC|B}(\rho).\label{eq:trignale-2}
\end{equation}

In the final step we note that the state (\ref{sigmaprime}) is PPT
with respect to all bipartitions, i.e., $E_{\mathrm{PPT}}^{AB|C}(\sigma')=E_{\mathrm{PPT}}^{AC|B}(\sigma')=E_{\mathrm{PPT}}^{A|BC}(\sigma')=0$.
The fact that $E_{\mathrm{PPT}}^{AB|C}(\sigma')=0$ is obvious, since
$\sigma'$ arises by performing a local von Neumann measurement $\{\Pi_{i}^{C}\}$
on the state $\sigma$, and thus has the form of a quantum-classical
state: $\sigma'=\sum_{i}\Pi_{i}^{C}\sigma\Pi_{i}^{C}=\sum_{i}p_{i}\sigma_{i}^{AB}\otimes\ket{i}\bra{i}^{C}$.
Moreover, by the very construction, the state $\sigma$ is PPT with
respect to the bipartition $AC|B$, and so is $\sigma'$, meaning
that $E_{\mathrm{PPT}}^{AC|B}(\sigma')=0$. This, together with the
fact that $\sigma'$ is classical on the subsystem $C$, implies that
it is also PPT with respect to the remaining bipartition $A|BC$:
$E_{\mathrm{PPT}}^{A|BC}(\sigma')=0$. This means that the distance
between $\rho$ and $\sigma'$ is an upper bound on $E_{\mathrm{PPT}}^{A|BC}(\rho)$,
which, when applied in Eq. (\ref{eq:trignale-2}), completes the proof. 
\end{proof}
The above theorem extends the range of applications of Eq. (\ref{eq:main})
to distance-based quantifiers of NPT entanglement. The same arguments
can also be applied to measures of distillability defined as \cite{Horodecki2009}
\begin{equation}
E_{\mathrm{ND}}^{X|Y}(\rho^{XY})=\min_{\sigma^{XY}\in\mathrm{ND}}D(\rho^{XY},\sigma^{XY}).
\end{equation}
Here, $\mathrm{ND}$ is the set of nondistillable states, and, as
before, the distance $D$ satisfies Eqs. (\ref{eq:contractive}) and
(\ref{eq:triangle}). Using the same arguments as in the proof of
Theorem \ref{thm:2} we see that Eq. (\ref{eq:PPT}) generalizes to
these distillability measures: 
\begin{equation}
\Delta^{C|AB}(\rho)\geq E_{\mathrm{ND}}^{A|BC}(\rho)-E_{\mathrm{ND}}^{AC|B}(\rho),\label{eq:ND}
\end{equation}
where the quantum discord $\Delta^{C|AB}$ is defined in the same
way as in Eq. (\ref{eq:discord-1}).

Finally, the above results also hold for measures of NPT entanglement
and distillability based on the relative entropy, although, as mentioned
earlier, the latter does not satisfy the triangle inequality in general.
The fact that Eqs. (\ref{eq:PPT}) and (\ref{eq:ND}) still apply
to these measures can be seen using the same arguments as in the proof
of Theorem \ref{thm:2} by observing that for the relative entropy
the inequality (\ref{eq:triangle-1}) becomes equality \cite{Streltsov2012}.

\subsection{Relation to Schatten norms}

The results presented so far hold for a very general class of quantifiers
for entanglement and discord. In particular, we have seen that Eq.
(\ref{eq:main}) applies for any entanglement measure $E$ which is
defined via the minimal distance to the set of separable, nondistillable,
or PPT states, if the amount of discord $\Delta$ is quantified as
in Eq. (\ref{eq:discord}). The corresponding distance only needs
to satisfy two minimal requirements given in Eqs. (\ref{eq:contractive})
and (\ref{eq:triangle}): it should not increase under quantum operations
and it should satisfy the triangle inequality. On the other hand,
we have also seen that Eq. (\ref{eq:main}) can be still valid even
if the distance violates one of these properties. This was demonstrated
for the relative entropy which can violate the triangle inequality.

In the following we will show that Eq. (\ref{eq:main}) may also hold
for distances violating Eq. (\ref{eq:contractive}), i.e., those that
are not contractive under quantum operations. To this end we will
consider the following distance 
\begin{equation}
D_{p}(\rho,\sigma)=||\rho-\sigma||_{p}
\end{equation}
with $\|\cdot\|_{p}$ being the Schatten $p$-norm of an operator
$M$ defined through 
\begin{equation}
||M||_{p}=\left(\mathrm{Tr}\left[\left(M^{\dagger}M\right)^{p/2}\right]\right)^{1/p}
\end{equation}
with $p\geq1$. Clearly, $D_{1}$ coincides with the trace distance
and thus does not increase under quantum operations \cite{Nielsen2000}.
However, contrary to what had been claimed in \cite{Witte1999}, already
$D_{2}$ (so-called Hilbert-Schmidt distance) can increase under quantum
operations as shown in \cite{Ozaw2000}. The arguments from Ref. \cite{Ozaw2000}
can be further generalized to show this fact for any $p>1$ (see also
\cite{Perez-Garcia2006,Piani2012} for similar considerations). 

Let now $E_{p}$ be defined as 
\begin{equation}
E_{p}^{X|Y}(\rho^{XY})=\min_{\sigma^{XY}\in T}D_{p}(\rho^{XY},\sigma^{XY}),
\end{equation}
with the minimization going over the set $T$, which here might denote
either of the sets: separable, nondistillable, or PPT states. Let
further $\Delta_{p}$ be defined by Eq. (\ref{eq:discord}) with the
distance taken to be $D_{p}$. The following theorem shows that Eq.
(\ref{eq:main}) also holds in this situation. 
\begin{thm}
\label{thm:3}Given a tripartite state $\rho=\rho^{ABC}$, the following
inequality holds: 
\begin{equation}
\Delta_{p}^{C|AB}(\rho)\geq E_{p}^{A|BC}(\rho)-E_{p}^{AC|B}(\rho).
\end{equation}
\end{thm}
\begin{proof}
The proof follows exactly the lines of the proof of Theorem \ref{thm:2}.
The only thing which needs to be proved is the fact that although
$D_{p}$ may increase under general quantum operations, it does not
for those operations that map the states $\rho$ and $\sigma$ to
$\rho'$ and $\sigma'$ in Eqs. (\ref{rhoprime}) and (\ref{sigmaprime}),
respectively. For this purpose, we notice that such mapping is unital,
i.e., $\sum_{i}\Pi_{i}^{C}\openone\Pi_{i}^{C}=\openone,$ where $\openone=\openone^{ABC}$
is the identity operator, and it was shown in Ref. \cite{Perez-Garcia2006}
that no unital map can increase the $p$-norm for any $p\geq1$. This
implies that for $p\geq1$, $D_{p}$ does satisfy Eq. (\ref{eq:contractive-1})
for the states of interest, which completes the proof. 
\end{proof}
While quantifiers of discord based on Schatten norms have been considered
only recently \cite{Dakic2010,Luo2010,Dakic2012,Girolami2012,Hassan2012,Bellomo2012,Piani2012,Gharibian2012,Vinjanampathy2012,Jin2012,Giampaolo2013,Paula2013},
entanglement quantifiers of this type were studied already more than
a decade ago \cite{Witte1999,Ozaw2000,Verstraete2002,Bertlmann2002,Bertlmann2005}.
Despite this fact, it has been an open question if $E_{p}$ is a proper
entanglement measure, i.e., if it is nonincreasing under LOCC for
$p>1$ \cite{Ozaw2000}. In what follows we will put this question
to rest by showing that $E_{p}$ can increase by simply discarding
a part of the system. For this purpose, let $\rho^{AB}$ be a quantum
state such that $E_{p}^{A|B}(\rho^{AB})>0$. Then, consider its extension
to a three-partite state defined as $\rho^{ABC}=\rho^{AB}\otimes\openone^{C}/2$,
where the particle $C$ is a qubit. We will now show that the entanglement
of $\rho^{AB}$ is larger than the entanglement of $\rho^{ABC}$:
\begin{equation}
E_{p}^{A|B}(\rho^{AB})>E_{p}^{A|BC}(\rho^{ABC})\label{eq:HS-3}
\end{equation}
for all $p>1$. To this end, observe that the amount of entanglement
$E_{p}^{A|BC}(\rho^{ABC})$ is bounded from above by the distance
$D_{p}(\rho^{ABC},\sigma^{ABC})$ for $\sigma^{ABC}=\sigma^{AB}\otimes\openone^{C}/2$,
where $\sigma^{AB}$ is the closest separable state to $\rho^{AB}$.
Moreover, notice that the distance between $\rho^{ABC}$ and $\sigma^{ABC}$
can also be expressed as \cite{Paula2013} 
\begin{equation}
D_{p}(\rho^{ABC},\sigma^{ABC})=\left\Vert \frac{\openone^{C}}{2}\right\Vert _{p}D_{p}(\rho^{AB},\sigma^{AB}).
\end{equation}
Recalling that the state $\sigma^{AB}$ was defined to be the closest
separable state to $\rho^{AB}$ and using the fact that $||\openone^{C}/2||_{p}=2^{1/p-1}$,
one obtains 
\begin{equation}
E_{p}^{A|BC}(\rho^{ABC})\leq2^{1/p-1}E_{p}^{A|B}(\rho^{AB}).
\end{equation}
The inequality (\ref{eq:HS-3}) follows by noting that for $p>1$
the exponent $1/p-1$ is negative, and thus $2^{1/p-1}<1$ in this
case.

Similar results with respect to quantum discord were also obtained
recently \cite{Piani2012,Paula2013}. In particular, it was pointed
out in \cite{Piani2012} that the geometric discord $\Delta_{\mathrm{G}}^{X|Y}=(\Delta_{2}^{X|Y})^{2}$
can increase under local operations on any of the parties $X$ or
$Y$, while most quantifiers of discord known in the literature do
not increase under quantum operations on the subsystem $Y$. This
result was later extended to all measures of discord $\Delta_{p}$
for $p>1$ \cite{Paula2013}. On the one hand, this observation together
with Eq. (\ref{eq:HS-3}) provides strong constraints for the possible
applications of entanglement and discord quantifiers based on Schatten
norms. On the other hand, the close relation of $E_{2}$ to the problem
of finding optimal entanglement witnesses \cite{Bertlmann2002,Bertlmann2005}
and the connection between $E_{p}$ and $\Delta_{p}$ established
in Theorem \ref{thm:3} demonstrate the use of these quantities for
understanding the structure of entanglement from a geometric perspective.

\section{\label{sec:Noisy}Noisy entanglement distribution}

In the scenario considered so far, Alice and Bob aimed at distributing
entanglement by having access to a noiseless quantum channel. Since
noise is unavoidable in any realistic experiment, we will now consider
the more general situation in which the channel used for entanglement
distribution is noisy. Similarly to the foregoing discussion we assume
that Alice and Bob have access to a tripartite initial state $\rho_{i}=\rho^{ABC}$,
where Alice is initially in possession of the particles $A$ and $C$,
and Bob is in possession of the remaining particle $B$. If Alice
uses a noisy channel $\Lambda^{C}$ to send her particle $C$ to Bob,
they end up in the final state $\rho_{f}=\Lambda^{C}[\rho_{i}]$.
The amount of entanglement distributed in this process is then given
by $E^{A|BC}(\rho_{f})-E^{AC|B}(\rho_{i})$.

Having introduced the concept of noisy entanglement distribution,
we are now in position to extend Eq. (\ref{eq:main}) to this general
scenario. In the following theorem we will show that noisy entanglement
distribution is in general limited by the amount of discord in each
of the states $\rho_{i}$ and $\rho_{f}$. 
\begin{thm}
\label{thm:noisy}Given a quantum channel $\Lambda^{C}$ and two states
$\rho_{i}=\rho^{ABC}$ and $\rho_{f}=\Lambda^{C}[\rho_{i}]$, the
following inequality holds: 
\begin{equation}
\min\{\Delta^{C|AB}(\rho_{i}),\Delta^{C|AB}(\rho_{f})\}\geq E^{A|BC}(\rho_{f})-E^{AC|B}(\rho_{i}).\label{eq:theorem4}
\end{equation}
Here, $E$ and $\Delta$ are any quantifiers of entanglement and discord
which satisfy Eq. (\ref{eq:main}). \end{thm}
\begin{proof}
We first apply Eq. (\ref{eq:main}) to the state $\rho_{i}$, thus
arriving at $\Delta^{C|AB}(\rho_{i})\geq E^{A|BC}(\rho_{i})-E^{AC|B}(\rho_{i})$.
Then, the inequality $\Delta^{C|AB}(\rho_{i})\geq E^{A|BC}(\rho_{f})-E^{AC|B}(\rho_{i})$
follows by recalling that entanglement does not increase under local
noise, i.e., $E^{A|BC}(\rho_{i})\geq E^{A|BC}(\rho_{f})$. Using analogous
reasoning one can also prove the inequality $\Delta^{C|AB}(\rho_{f})\geq E^{A|BC}(\rho_{f})-E^{AC|B}(\rho_{i})$.
Application of Eq. (\ref{eq:main}) to the state $\rho_{f}$ gives
us the inequality $\Delta^{C|AB}(\rho_{f})\geq E^{A|BC}(\rho_{f})-E^{AC|B}(\rho_{f})$.
One then completes the proof by using $E^{AC|B}(\rho_{f})\leq E^{AC|B}(\rho_{i})$,
which again follows from the fact that entanglement does not increase
under local noise. As quantum discord can increase or decrease under
local noise \cite{Dakic2010,Streltsov2011,Campbell2011,Ciccarello2012},
the claim follows. 
\end{proof}

\subsection{Divisible channels}

\begin{figure}
\includegraphics[width=1\columnwidth]{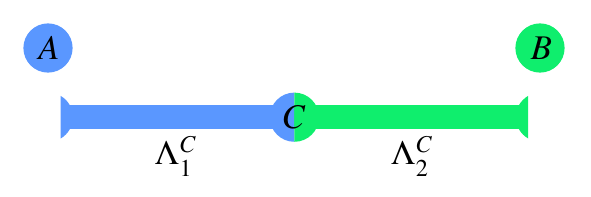}

\caption{\label{fig:channel}Decomposition of a noisy channel $\Lambda^{C}$
in two channels $\Lambda_{1}^{C}$ and $\Lambda_{2}^{C}$. For an
initial state $\rho_{i}=\rho^{ABC}$ the final state after the application
of the channel is given by $\rho_{f}=\Lambda^{C}[\rho_{i}]=\Lambda_{2}^{C}[\Lambda_{1}^{C}[\rho_{i}]]$.
The figure illustrates the intermediate state $\tilde{\rho}=\Lambda_{1}^{C}[\rho_{i}]$
after the application of $\Lambda_{1}^{C}$ only. See main text for
details.}
\end{figure}
Let us consider a decomposition of the channel $\Lambda^{C}$ into
two channels $\Lambda_{1}^{C}$ and $\Lambda_{2}^{C}$ such that the
successive application of these channels is equivalent to the application
of $\Lambda^{C}$: 
\begin{equation}
\rho_{f}=\Lambda^{C}[\rho_{i}]=\Lambda_{2}^{C}[\Lambda_{1}^{C}[\rho_{i}]];
\end{equation}
see also Fig. \ref{fig:channel} for an illustration. If such a decomposition
is possible with nonunitary $\Lambda_{1}^{C}$ and $\Lambda_{2}^{C}$,
the channel $\Lambda^{C}$ is called divisible \cite{Wolf2008}. By
introducing an intermediate state 
\begin{equation}
\tilde{\rho}=\Lambda_{1}^{C}[\rho_{i}]
\end{equation}
we will now show that the amount of distributed entanglement is in
general bounded above by the amount of discord in the state $\tilde{\rho}$:

\begin{equation}
\Delta^{C|AB}(\tilde{\rho})\geq E^{A|BC}(\rho_{f})-E^{AC|B}(\rho_{i}).\label{eq:noisy-3}
\end{equation}
As in the foregoing discussion, we assume that $E$ and $\Delta$
are quantifiers of entanglement and discord satisfying Eq. (\ref{eq:main}).
Under this assumption, Eq. (\ref{eq:noisy-3}) can be proven using
similar arguments as in the proof of Eq. (\ref{eq:theorem4}). In
particular, we can apply Eq. (\ref{eq:main}) to the intermediate
state $\tilde{\rho}$, thus obtaining the inequality $\Delta^{C|AB}(\tilde{\rho})\geq E^{A|BC}(\tilde{\rho})-E^{AC|B}(\tilde{\rho})$.
The proof of Eq. (\ref{eq:noisy-3}) is complete by making use of
the fact that entanglement does not increase under local noise, leading
to the inequalities $E^{A|BC}(\tilde{\rho})\geq E^{A|BC}(\rho_{f})$
and $E^{AC|B}(\tilde{\rho})\leq E^{AC|B}(\rho_{i})$.

\subsection{Markovian time evolution}

Here we will see that the results presented in the previous section
have a nice application in the scenario in which the particle $C$
used for entanglement distribution is subject to a Markovian time
evolution $\Lambda_{(t_{2},t_{1})}^{C}$. If we assume that the process
starts with the initial state $\rho_{i}=\rho^{ABC}$ at the time $t=0$,
then for any time $t\geq0$ the time-evolved state is given by 
\begin{equation}
\rho_{t}=\Lambda_{(t,0)}^{C}[\rho_{i}].\label{eq:rhot}
\end{equation}
Denoting then by $T$ the total time required for the process, the
corresponding final state $\rho_{f}$ can be written as 
\begin{equation}
\rho_{f}=\rho_{T}=\Lambda_{(T,0)}^{C}[\rho_{i}].\label{eq:rhof}
\end{equation}

We are now in position to prove that the amount of entanglement distributed
via a Markovian time evolution is bounded from above by the amount
of discord in the time-evolved state $\rho_{t}$ for any $T\geq t\geq0$:
\begin{equation}
\Delta^{C|AB}(\rho_{t})\geq E^{A|BC}(\rho_{f})-E^{AC|B}(\rho_{i}).\label{eq:Markovian}
\end{equation}
Here, $E$ and $\Delta$ are quantifiers of entanglement and discord
satisfying Eq. (\ref{eq:main}). To prove the above statement we use
the fact that any Markovian time evolution $\Lambda_{(t_{2},t_{1})}^{C}$
obeys the composition law \cite{Rivas2014}, that is, 
\begin{equation}
\Lambda_{(t_{2},t_{1})}^{C}[\rho]=\Lambda_{(t_{2},t)}^{C}\left[\Lambda_{(t,t_{1})}^{C}[\rho]\right]
\end{equation}
for any state $\rho$ and all $t_{2}\geq t\geq t_{1}\geq0$. This
together with Eqs. (\ref{eq:rhot}) and (\ref{eq:rhof}) lead us to
the following expression for the final state 
\begin{equation}
\rho_{f}=\Lambda_{(T,0)}^{C}[\rho_{i}]=\Lambda_{(T,t)}^{C}[\rho_{t}]
\end{equation}
for all $T\geq t\geq0$. One then obtains Eq. (\ref{eq:Markovian})
by applying Eq. (\ref{eq:noisy-3}) with $\tilde{\rho}=\rho_{t}$.

Let us notice that the inequality (\ref{eq:Markovian}) also implies
that the distribution of entanglement via a Markovian time evolution
is bounded above by the minimal discord $\min_{t}\Delta^{C|AB}(\rho_{t})$,
minimized over all times $t$ ranging between $0$ and the duration
of the total procedure $T$. On the other hand, any violation of Eq.
(\ref{eq:Markovian}) can also be regarded as a witness for the non-Markovianity
of the underlying time evolution. These results support the recent
attempts to detect and quantify non-Markovianity via quantum entanglement
\cite{Rivas2010} and quantum discord \cite{Girolami2012,Alipour2012,Haikka2013}.
Noting that the inequality (\ref{eq:Markovian}) is valid for a very
general class of quantifiers for entanglement and discord, further
investigation in this direction can lead to a better understanding
of entanglement and discord in the context of detecting non-Markovianity.

\section{\label{sec:Optimal}Optimal entanglement distribution without preshared
correlations}

In the foregoing discussion we considered noiseless and noisy entanglement
distribution, and presented several tools for bounding the amount
of entanglement distributed in this process. In this section we will
apply them to the following problem: \textit{How much entanglement
can be distributed via a given quantum channel?}

Let us begin with the scenario in which Alice and Bob are not correlated
initially, i.e, the initial and the final state are given by 
\begin{align}
\rho_{i}=\rho^{AC}\otimes\rho^{B}\label{eq:initial}
\end{align}
and 
\begin{align}
\rho_{f}=\Lambda^{C}[\rho^{AC}]\otimes\rho^{B}
\end{align}
respectively. We assume again that Alice is initially in possession
of the particles $A$ and $C$, while Bob holds the particle $B$.
In the distribution process, the particle $C$ is sent from Alice
to Bob via the quantum channel $\Lambda^{C}$. Thus, the initial entanglement
between Alice and Bob is zero, and the amount of distributed entanglement
is given by $E^{A|C}(\Lambda^{C}[\rho^{AC}])$. 

In the following, we are interested in optimal entanglement distribution,
i.e., we ask which initial states $\rho^{AC}$ lead to the maximal
final entanglement after the application of a quantum channel $\Lambda^{C}$.
Clearly, if the quantum channel $\Lambda^{C}$ is noiseless, the optimal
distribution strategy is achieved if Alice prepares her particles
$A$ and $C$ in the maximally entangled state 
\begin{equation}
\ket{\phi^{+}}^{AC}=\frac{1}{\sqrt{d_{C}}}\sum_{i=0}^{d_{C}-1}\ket{ii}^{AC}\label{MaxEnt}
\end{equation}
and sends the particle $C$ to Bob.

Interestingly, as we will see below, this strategy is not always optimal
if the quantum channel $\Lambda^{C}$ is noisy. In passing, it is
crucial to notice that all maximally entangled states show the same
performance for entanglement distribution, i.e., 
\begin{equation}
E(\Lambda^{C}[\proj{\phi_{\mathrm{me}}}^{AC}])=E(\Lambda^{C}[\proj{\phi^{+}}^{AC}])\label{eq:maximaly entangled}
\end{equation}
is true for any maximally entangled state $\ket{\phi_{\mathrm{me}}}^{AC}$,
any entanglement measure $E$, and any noisy channel $\Lambda^{C}$.
This can be seen by first noting that any maximally entangled state
$\ket{\phi_{\mathrm{me}}}^{AC}$ can be written as $\ket{\phi_{\mathrm{me}}}^{AC}=U_{A}\ket{\phi^{+}}^{AC}$,
where $U_{A}$ is a unitary acting on the subsystem $A$. Then, to
get Eq. (\ref{eq:maximaly entangled}) one uses the facts that $U_{A}$
commutes with $\Lambda^{C}$ and that any entanglement quantifier
$E$ is invariant under local unitaries \cite{Horodecki2009}.

\subsection{\label{sub:EoF}Relation to entanglement of formation}

In this section we will show that maximally entangled states are optimal
for entanglement distribution for all noisy channels if the exchanged
particle $C$ is a qubit and the amount of entanglement is quantified
via the entanglement of formation $E_{f}$. We then have the following
theorem.
\begin{thm}
\label{thm:EoF}For any mixed state $\rho^{AC}$ with $d_{A}\geq d_{C}=2$,
and any channel $\Lambda^{C}$ the following inequality holds: 
\begin{equation}
E_{f}(\Lambda^{C}[\proj{\phi^{+}}^{AC}])\geq E_{f}(\Lambda^{C}[\rho^{AC}]).
\end{equation}
\end{thm}
\begin{proof}
We first recall that the entanglement of formation is a convex function
of the state. This implies that for any mixed state $\rho^{AC}$ there
exists a pure state $\ket{\psi}^{AC}$ which shows at least the same
performance for entanglement distribution: 
\begin{equation}
E_{f}(\Lambda^{C}[\proj{\psi}^{AC}])\geq E_{f}(\Lambda^{C}[\rho^{AC}]).
\end{equation}
To complete the proof we will show that the maximally entangled state
has the best performance among all pure states, i.e., 
\begin{equation}
E_{f}(\Lambda^{C}[\proj{\phi^{+}}^{AC}])\geq E_{f}(\Lambda^{C}[\proj{\psi}^{AC}])
\end{equation}
for any pure state $\ket{\psi}^{AC}$ with $d_{A}\geq2$, $d_{C}=2$,
and any single-qubit channel $\Lambda^{C}$. At this point, it is
important to note that the state $\ket{\psi}^{AC}$ is effectively
a two-qubit state, even if the dimension of the subsystem $A$ is
larger than two. This follows from the Schmidt decomposition of $\ket{\psi}^{AC}$,
which due to the fact that the subsystem $C$ is two-dimensional,
is of the form $\ket{\psi}^{AC}=\lambda_{0}\ket{00}+\lambda_{1}\ket{11}$.
The state $\Lambda^{C}[\ket{\psi}\bra{\psi}^{AC}]$ can thus be regarded
as a mixed state of two qubits. With this in mind, we can now use
the fact that for all two-qubit states the entanglement of formation
admits a simple formula: $E_{f}=g(\mathcal{C})$, where $g$ is a
nondecreasing function and $\mathcal{C}$ is the concurrence \cite{Wootters1998}.
The final ingredient of our proof is the factorization law for concurrence
(see Eq. (5) in \cite{Konrad2008}). Adapted to our notation it reads
\begin{equation}
\mathcal{C}(\Lambda^{C}[\proj{\psi}^{AC}])=\mathcal{C}(\Lambda^{C}[\proj{\phi^{+}}^{AC}])\cdot\mathcal{C}(\proj{\psi}^{AC}).
\end{equation}
Since the concurrence is never larger than one, we arrive at the following
inequality: 
\begin{equation}
\mathcal{C}(\Lambda^{C}[\proj{\phi^{+}}^{AC}])\geq\mathcal{C}(\Lambda^{C}[\proj{\psi}^{AC}]).
\end{equation}
Note that this inequality also holds if the concurrence $\mathcal{C}$
is replaced by the entanglement of formation $E_{f}$, since the latter
is a nondecreasing function of the concurrence. This observation completes
the proof %
\footnote{Note that in the proof of Theorem \ref{thm:EoF} we used the fact
that the entanglement of formation is convex. In particular, convexity
was used to prove the existence of a pure state which is optimal for
entanglement distribution. It is also possible to prove the theorem
without referring to convexity. This can be seen by noting that for
any mixed state $\rho^{AC}$ there exists a purification $\ket{\psi^{RAC}}$
such that $\mathrm{Tr}_{R}[\ket{\psi}\bra{\psi}^{RAC}]=\rho^{AC}$.
Since entanglement does not increase under discarding systems, the
purification $\ket{\psi^{RAC}}$ shows at least the same performance
for entanglement distribution as the state $\rho^{AC}$: $E^{RA|C}(\Lambda^{C}[\ket{\psi}\bra{\psi}^{RAC}])\geq E^{A|C}(\Lambda^{C}[\rho^{AC}])$.%
}. 
\end{proof}
It is worth mentioning that the above result can be generalized to
a larger class of entanglement measures, namely to all those measures
which for two qubits can be written as a nondecreasing function of
concurrence, that is, 
\begin{equation}
E=g(\mathcal{C}).
\end{equation}
This can be seen by exploiting the same argumentation as before. Apart
from the entanglement of formation, examples of such measures are
the geometric measure of entanglement \cite{Shimony1995,Wei2003},
the Bures measure of entanglement \cite{Vedral1997,Vedral1998}, and
the Groverian measure of entanglement \cite{Biham2002,Shapira2006}.
For two qubits all those measures reduce to a nondecreasing function
of concurrence, see Fig. 4 in ref. \cite{Streltsov2010}.

\subsection{Relation to Pauli channel\label{sub:Pauli}}

We now show that for an important type of noise - the Pauli channel
- the statement made in the previous section can be generalized to
all entanglement measures. The action of the Pauli channel reads:
\begin{equation}
\Lambda_{\mathrm{p}}^{C}[\rho^{AC}]=\sum_{i=0}^{3}p_{i}\sigma_{i}^{C}\rho^{AC}\sigma_{i}^{C},
\end{equation}
where the exchanged particle $C$ is a qubit and $\sigma_{i}$ are
Pauli matrices with $\sigma_{0}=\openone$. We have the following
\begin{figure}
\includegraphics[width=1\columnwidth]{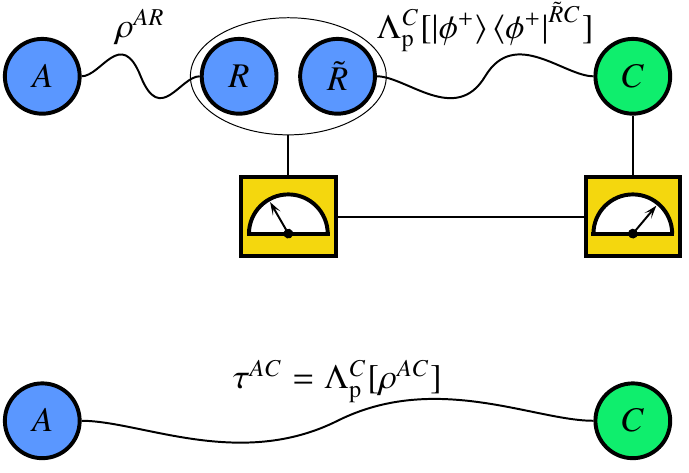}

\caption{\label{fig:teleportation}The state $\Lambda_{\mathrm{p}}^{C}[\ket{\phi^{+}}\bra{\phi^{+}}^{\tilde{R}C}]$
can be used to teleport the particle $R$ of the state $\rho^{AR}$
by performing a joint Bell measurement on $R$ and $\tilde{R}$, and
a conditional rotation on $C$ (upper figure). This procedure leaves
the subsystem $AC$ in the final state $\tau^{AC}=\Lambda_{\mathrm{p}}^{C}[\rho^{AC}]$
(lower figure).}
\end{figure}

\begin{thm}
\label{thm:Pauli}For any mixed state $\rho^{AC}$ with $d_{A}\geq d_{C}=2$
and any Pauli channel $\Lambda_{\mathrm{p}}^{C}$ the following inequality
holds: 
\begin{equation}
E(\Lambda_{\mathrm{p}}^{C}[\proj{\phi^{+}}^{AC}])\geq E(\Lambda_{\mathrm{p}}^{C}[\rho^{AC}])
\end{equation}
for any entanglement measure $E$.\end{thm}
\begin{proof}
Let us start by introducing two additional particles $R$ and $\tilde{R}$
with $d_{R}=d_{\tilde{R}}=2$. We will now show that the state $\Lambda_{\mathrm{p}}^{C}[\proj{\phi^{+}}^{\tilde{R}C}]$
can be used for teleportation in the following way: if two parties
share the state $\Lambda_{\mathrm{p}}^{C}[\proj{\phi^{+}}^{\tilde{R}C}]$
and apply the standard teleportation protocol \cite{Bennett1993}
for teleporting the two-dimensional subsystem $R$ of a total state
$\rho^{AR}$, they will end up sharing the state $\Lambda_{\mathrm{p}}^{C}[\rho^{AC}]$,
see also \cite{Bowen2001} for similar considerations. This can be
seen explicitly by considering the essential steps of the standard
teleportation protocol (see Fig. \ref{fig:teleportation}). In the
first step, a joint Bell measurement is performed on the subsystems
$R$ and $\tilde{R}$. Depending on the outcome $i$ of the measurement,
the subsystem $AC$ is found in one of the four states $\Lambda_{\mathrm{p}}^{C}[\sigma_{i}^{C}\rho^{AC}\sigma_{i}^{C}]$
with $0\leq i\leq3$. In the final step, a conditioned unitary rotation
$\sigma_{i}^{C}$ is applied on the subsystem $C$, leading to the
final state 
\begin{equation}
\tau^{AC}=\sigma_{i}^{C}\Lambda_{\mathrm{p}}^{C}[\sigma_{i}^{C}\rho^{AC}\sigma_{i}^{C}]\sigma_{i}^{C}.
\end{equation}
At this point, it is crucial to note that the Pauli channel commutes
with the Pauli matrices $\sigma_{i}^{C}$, i.e., 
\begin{equation}
\Lambda_{\mathrm{p}}^{C}[\sigma_{i}^{C}\rho^{AC}\sigma_{i}^{C}]=\sigma_{i}^{C}\Lambda_{\mathrm{p}}^{C}[\rho^{AC}]\sigma_{i}^{C},\label{eq:commutation}
\end{equation}
which can be seen by inspection using the anticommutation relation
$\sigma_{a}\sigma_{b}=-\sigma_{b}\sigma_{a}$ for $1\leq a,b\leq3$.
Using Eq. (\ref{eq:commutation}) we see that the final state $\tau^{AC}$
becomes independent from the outcome of the measurement $i$: 
\begin{equation}
\tau^{AC}=\Lambda_{\mathrm{p}}^{C}[\rho^{AC}].
\end{equation}
Finally, note that all steps mentioned above can be performed by using
local operations and classical communication (see Fig. \ref{fig:teleportation}).
This implies that the final state $\tau^{AC}$ cannot have more entanglement
than the state $\Lambda_{\mathrm{p}}^{C}[\proj{\phi^{+}}^{\tilde{R}C}]$,
regardless of the entanglement measure $E$ used to quantify it. This
completes the proof. 
\end{proof}
It should be stressed that the result presented in Theorem \ref{thm:Pauli}
can also be extended to the scenario in which the channel used for
entanglement distribution is a tensor product of different single-qubit
Pauli channels. As an example, consider a four dimensional particle
$C$ consisting of two qubits $C_{1}$ and $C_{2}$. The channel $\Lambda_{\mathrm{p}}^{C}$
is now of the form $\Lambda_{\mathrm{p}}^{C}=\Lambda_{\mathrm{p}}^{C_{1}}\otimes\widetilde{\Lambda}_{\mathrm{p}}^{C_{2}}$,
where $\Lambda_{\mathrm{p}}^{C_{1}}$ and $\widetilde{\Lambda}_{\mathrm{p}}^{C_{2}}$
are two (possibly different) Pauli channels. The action of this channel
onto an arbitrary state $\rho^{AC}=\rho^{AC_{1}C_{2}}$ is given by
\begin{equation}
\Lambda_{\mathrm{p}}^{C}[\rho^{AC}]=\Lambda_{\mathrm{p}}^{C_{1}}\otimes\widetilde{\Lambda}_{\mathrm{p}}^{C_{2}}[\rho^{AC_{1}C_{2}}].
\end{equation}
Using similar lines of reasoning as in the proof of Theorem \ref{thm:Pauli}
we see that the best performance in this case is also achieved for
the maximally entangled state, i.e., the inequality 
\begin{equation}
E(\Lambda_{\mathrm{p}}^{C}[\proj{\phi^{+}}^{AC}])\geq E(\Lambda_{\mathrm{p}}^{C}[\rho^{AC}])
\end{equation}
holds for any state $\rho^{AC}$ with $d_{A}\geq d_{C}=4$ and the
maximally entangled state $\ket{\phi^{+}}^{AC}=(1/2)\sum_{i=0}^{3}\ket{ii}^{AC}$.
This statement is also true if the exchanged particle $C$ consists
of $n$ qubits, and the channel $\Lambda_{\mathrm{p}}^{C}$ is a combination
of $n$ (possibly different) single-qubit Pauli channels. In this
case, the best performance is achieved for the maximally entangled
state (\ref{MaxEnt}) with $d_{A}\geq d_{C}=2^{n}$.

Finally, we note that similar arguments can also be applied to a more
general family of channels defined as follows: 
\begin{equation}
\Lambda^{C}[\rho^{AC}]=\sum_{i}p_{i}U_{i}^{C}\rho^{AC}\left(U_{i}^{C}\right)^{\dagger},\label{eq:channel}
\end{equation}
where the particles $A$ and $C$ can have arbitrary dimensions and
$U_{i}^{C}$ are unitary operators that act only on the particle $C$
and have the following two properties: 
\begin{itemize}
\item The unitaries $\left(U_{i}^{C}\right)^{\dagger}$ commute with the
channel $\Lambda^{C}$, i.e.,
\begin{equation}
\Lambda^{C}\left[\left(U_{i}^{C}\right)^{\dagger}\rho^{AC}U_{i}^{C}\right]=\left(U_{i}^{C}\right)^{\dagger}\Lambda^{C}\left[\rho^{AC}\right]U_{i}^{C},\label{eq:commutation-1}
\end{equation}

\item For the maximally entangled state $\ket{\phi^{+}}^{AC}$, the states
\begin{equation}
\ket{\psi_{i}}^{AC}=U_{i}^{C}\ket{\phi^{+}}^{AC}
\end{equation}
form a complete orthonormal basis, i.e., $\braket{\psi_{i}|\psi_{j}}=\delta_{ij}$
and $\sum_{i}\ket{\psi_{i}}\bra{\psi_{i}}=\openone$.
\end{itemize}
As we will show in the following, the maximally entangled state $\ket{\phi^{+}}^{AC}$
is optimal for entanglement distribution via a noisy channel given
in Eq.~(\ref{eq:channel}). We will prove this statement by following
the same reasoning as for Pauli channels, see also Fig.~\ref{fig:teleportation}.
In particular, we will show that the state $\Lambda^{C}[\ket{\phi^{+}}\bra{\phi^{+}}^{\tilde{R}C}]$
can be used to teleport the particle $R$ of dimension not larger
than $d_{C}$, such that for any state $\rho^{AR}$ the final state
has the form $\tau^{AC}=\Lambda^{C}[\rho^{AC}]$. This can be proven
by considering the state $\rho^{AR}\otimes\Lambda^{C}[\ket{\phi^{+}}\bra{\phi^{+}}^{\tilde{R}C}]$,
and applying a joint measurement on the particles $R$ and $\tilde{R}$
in the basis $\ket{\psi_{i}}^{R\tilde{R}}=U_{i}^{R}\ket{\phi^{+}}^{R\tilde{R}}$.
Conditioned on the measurement outcome $i$, the resulting post-measurement
state of the particles $A$ and $C$ is then given by $\Lambda^{C}[(U_{i}^{C})^{\dagger}\rho^{AC}U_{i}^{C}]$.
In the final step of the proof, we use Eq.~(\ref{eq:commutation-1})
and apply conditional unitary rotations $U_{i}^{C}$, arriving at
the desired final state $\tau^{AC}=\Lambda^{C}[\rho^{AC}]$. Using
the same reasoning as for the Pauli channels, this proves the optimality
of the maximally entangled state $\ket{\phi^{+}}^{AC}$ for the channels
given in Eq.~(\ref{eq:channel}). Examples of such channels more
general than the Pauli channels are the Weyl covariant channels.

\subsection{\label{sub:negativity}Relation to negativity and amplitude damping
channel}

All the results presented so far support the intuition that sending
one half of a maximally entangled state down a noisy quantum channel
represents the optimal strategy if two parties wish to distribute
entanglement between them. In particular, we have seen that this statement
is true for all single-qubit channels if the figure of merit is the
entanglement of formation, or any other entanglement measure which
for two qubits reduces to a nondecreasing function of concurrence.
Moreover, for Pauli channels we saw that this statement becomes completely
general: in this case maximally entangled states are the optimal resource,
regardless of the entanglement measure used.

\begin{figure}
\includegraphics[width=1\columnwidth]{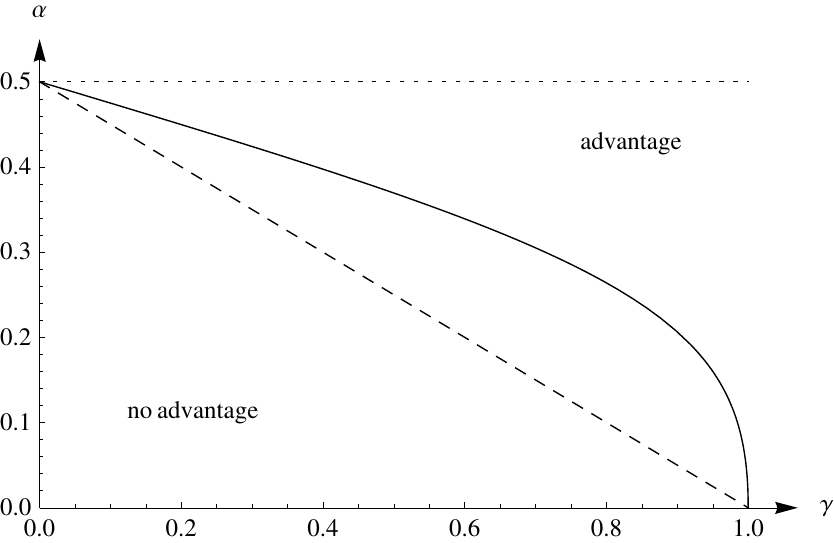}

\caption{\label{fig:amplitudedamping}The plot shows the relevant parameter
regions of the damping parameter $\gamma$ and the parameter $\alpha$
which enters the initial state $\ket{\alpha}^{AC}$, see Eq. (\ref{eq:alpha}).
Dashed line $\alpha=(1-\gamma)/2$ separates the parameter space in
two parts. For $(1-\gamma)/2<\alpha<1/2$ all states $\ket{\alpha}^{AC}$
outperform the maximally entangled state. The solid line shows the
value of $\alpha_{\max}$ which leads to the maximal logarithmic negativity
for a given damping parameter $\gamma$. The dotted line shows $\alpha=1/2$.}
\end{figure}
\begin{figure}
\includegraphics[width=1\columnwidth]{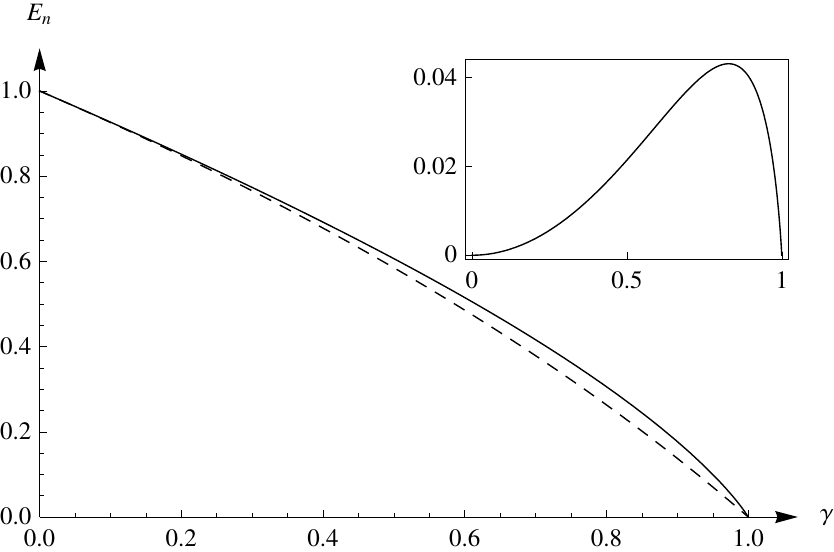}

\caption{\label{fig:amplitudedamping2}The solid line shows the logarithmic
negativity $E_{n}$ of the state $\Lambda_{\mathrm{ad}}^{C}[\proj{\alpha_{\max}}^{AC}]$
as function of the damping parameter $\gamma$. Then, the dashed line
is the corresponding logarithmic negativity of $\Lambda_{\mathrm{ad}}^{C}[\proj{\phi^{+}}^{AC}]$.
The inset shows the difference $E_{n}(\Lambda_{\mathrm{ad}}^{C}[\proj{\alpha_{\max}}^{AC}])-E_{n}(\Lambda_{\mathrm{ad}}^{C}[\proj{\phi^{+}}^{AC}])$.
The states $\ket{\alpha_{\max}}^{AC}$ outperform the maximally entangled
state $\ket{\phi^{+}}^{AC}$ in the whole region $0<\gamma<1$.}
\end{figure}
Quite surprisingly, this intuition is in general not correct \cite{Ziman2007,Pal}.
In particular, it was shown in Ref. \cite{Pal} (see Sec. III therein)
that maximally entangled states are not optimal for entanglement distribution
if the amount of entanglement is quantified by the negativity \cite{Ziyczkowski1998,Vidal2002}
which is defined as $N(\rho^{AC})=||\rho^{T_{A}}||_{1}-1$, where
$T_{A}$ denotes the partial transposition over the system $A$ and
$||M||_{1}=\mathrm{Tr}\sqrt{M^{\dagger}M}$ is the trace norm of $M$.
In what follows we will recall this result, using however a slightly
different entanglement monotone which is the logarithmic negativity
given by 
\begin{equation}
E_{n}(\rho^{AC})=\log_{2}||\rho^{T_{A}}||_{1}=\log_{2}(N(\rho^{AC})+1).
\end{equation}
We will also supplement the results of Ref. \cite{Pal} by noting
that for some quantum channels even arbitrarily little entangled states
can outperform the maximally entangled state.

The effect of suboptimality of maximally entangled states was demonstrated
for the single-qubit amplitude damping channel 
\begin{equation}
\Lambda_{\mathrm{ad}}^{C}[\rho^{AC}]=K_{1}\rho^{AC}K_{1}^{\dagger}+K_{2}\rho^{AC}K_{2}^{\dagger}\label{eq:amplitude damping}
\end{equation}
with Kraus operators $K_{1}=\ket{0}\!\bra{0}^{C}+\sqrt{1-\gamma}\ket{1}\!\bra{1}^{C}$
and $K_{2}=\sqrt{\gamma}\ket{0}\!\bra{1}^{C}$, and the damping parameter
$0\leq\gamma\leq1$. If the initial state is chosen as 
\begin{equation}
\ket{\alpha}^{AC}=\sqrt{1-\alpha}\ket{00}^{AC}+\sqrt{\alpha}\ket{11}^{AC}\label{eq:alpha}
\end{equation}
with the real parameter $0\leq\alpha\leq1$, it is straightforward
to verify that for $\tilde{\alpha}=(1-\gamma)/2$ the state $\ket{\tilde{\alpha}}^{AC}$
shows the same performance as the maximally entangled state, that
is, 
\begin{equation}
E_{n}(\Lambda_{\mathrm{ad}}^{C}[\proj{\tilde{\alpha}}^{AC}])=E_{n}(\Lambda_{\mathrm{ad}}^{C}[\proj{\phi^{+}}^{AC}]).
\end{equation}
This is illustrated in Fig. \ref{fig:amplitudedamping}, where the
parameter space of $\alpha$ and $\gamma$ is shown. The dashed line
for $\alpha=\tilde{\alpha}=(1-\gamma)/2$ divides the parameter space
into two parts. For $\alpha\leq(1-\gamma)/2$ (lower left triangle
in Fig. \ref{fig:amplitudedamping}) the state $\ket{\alpha}^{AC}$
shows no advantage when compared to the maximally entangled state,
i.e., 
\begin{equation}
E_{n}(\Lambda_{\mathrm{ad}}^{C}[\proj{\alpha}^{AC}])\leq E_{n}(\Lambda_{\mathrm{ad}}^{C}[\proj{\phi^{+}}^{AC}]).
\end{equation}
However, for $(1-\gamma)/2<\alpha<1/2$ (upper right triangle in Fig.
\ref{fig:amplitudedamping}) the corresponding state $\ket{\alpha}^{AC}$
always outperforms the maximally entangled state for the damping parameter
$0<\gamma<1$: 
\begin{equation}
E_{n}(\Lambda_{\mathrm{ad}}^{C}[\proj{\alpha}^{AC}])>E_{n}(\Lambda_{\mathrm{ad}}^{C}[\proj{\phi^{+}}^{AC}]).\label{eq:advantage}
\end{equation}

For a given damping parameter $\gamma$ we can further maximize the
logarithmic negativity of the state $\Lambda_{\mathrm{ad}}^{C}[\proj{\alpha}^{AC}]$
with respect to the parameter $\alpha$. Direct algebra shows that
the maximum is achieved for 
\begin{equation}
\alpha_{\max}=\frac{1}{\frac{\gamma}{\sqrt{1-\gamma}}+2}
\end{equation}
(see the solid line in Fig. \ref{fig:amplitudedamping}). The corresponding
quantity $E_{n}(\Lambda_{\mathrm{ad}}^{C}[\proj{\alpha_{\max}}^{AC}])$
is shown in Fig. \ref{fig:amplitudedamping2} as a function of the
damping parameter $\gamma$ (solid line). There, we also show the
logarithmic negativity $E_{n}(\Lambda_{\mathrm{ad}}^{C}[\proj{\phi^{+}}^{AC}])$
for the maximally entangled state (dashed line).

More interestingly, however, it turns out that for the logarithmic
negativity maximally entangled states can be outperformed even by
states with arbitrarily little entanglement, which we prove in the
following theorem. 
\begin{thm}
\label{thm:negativity}For any $\varepsilon>0$ there exists a state
$\rho_{\varepsilon}^{AC}$ with logarithmic negativity at most $\varepsilon$
and a channel $\Lambda_{\varepsilon}^{C}$ such that 
\begin{equation}
E_{n}(\Lambda_{\varepsilon}^{C}[\rho_{\varepsilon}^{AC}])>E_{n}(\Lambda_{\varepsilon}^{C}[\proj{\phi^{+}}^{AC}]).
\end{equation}
\end{thm}
\begin{proof}
We show this result for the amplitude damping channel $\Lambda_{\mathrm{ad}}^{C}$
given in Eq. (\ref{eq:amplitude damping}) and the pure state $\ket{\alpha}^{AC}$
in Eq. (\ref{eq:alpha}). From the fact that the state $\ket{\alpha}^{AC}$
is separable for $\alpha=0$ and maximally entangled for $\alpha=1/2$,
it follows that for any $\varepsilon>0$ there exists $\alpha_{\varepsilon}\in(0,1/2)$
such that the logarithmic negativity of the state $\ket{\alpha_{\varepsilon}}^{AC}$
is nonzero and at most $\varepsilon$, i.e., 
\begin{equation}
0<E_{n}(\ket{\alpha_{\varepsilon}}^{AC})\leq\varepsilon.
\end{equation}
To complete the proof it is enough to show that for any $\varepsilon>0$
there exists an amplitude damping channel $\Lambda_{\mathrm{ad}}^{C}$
with the damping parameter $\gamma_{\varepsilon}$ such that 
\begin{equation}
E_{n}(\Lambda_{\mathrm{ad}}^{C}[\proj{\alpha_{\varepsilon}}^{AC}])>E_{n}(\Lambda_{\mathrm{ad}}^{C}[\proj{\phi^{+}}^{AC}].\label{eq:alphaepsilon}
\end{equation}
The existence of such a channel follows directly from the arguments
presented above. Precisely, by virtue of the inequality (\ref{eq:advantage})
we immediately see that Eq. (\ref{eq:alphaepsilon}) is true for any
damping parameter $\gamma_{\varepsilon}$ chosen such that $1-2\alpha_{\varepsilon}<\gamma_{\varepsilon}<1$. 
\end{proof}
We have then shown that, in some scenarios, states with very little
entanglement are a better resource for noisy entanglement distribution
when compared to maximally entangled states if the logarithmic negativity
$E_{n}$ is used to quantify entanglement. It is worth mentioning
that this entanglement measure is closely related to the PPT entanglement
cost, i.e., the entanglement cost under quantum operations preserving
the positivity of the partial transpose. Precisely, $E_{n}$ is always
a lower bound on the PPT entanglement cost \cite{Audenaert2003},
while for all two-qubit states both quantities coincide \cite{Ishizaka2004}.
For this reason the logarithmic negativity is equivalent to the PPT
entanglement cost within the framework presented in this section,
and all statements made for one quantity are also valid for the other.

Moreover, we point out that the result presented in Theorem \ref{thm:negativity}
can also be extended to the multi-copy scenario, where Alice and Bob
have access to many copies of a quantum channel $\Lambda^{C}$. The
aim of the process in this case is to distribute the maximal logarithmic
negativity per copy of the channel. The aforementioned results together
with additivity of the logarithmic negativity \cite{Vidal2002} imply
that for amplitude damping noise maximally entangled states can be
outperformed by states with arbitrary little entanglement also in
this scenario.

Let us finally mention that in Ref. \cite{Pal} the authors show that
the maximally entangled states are optimal for entanglement distribution
if the single-qubit channel used to transmit the particle is unital
and negativity is used as the entanglement measure.

\section{\label{sec:Advantage}Optimal entanglement distribution with preshared
correlations}

In the foregoing discussion on optimal entanglement distribution we
assumed that initially Alice and Bob do not share any correlations,
i.e., the initial state is fully product, see Eq. (\ref{eq:initial}).
Here, we will relax this assumption and allow for more general initial
quantum states 
\begin{equation}
\rho_{i}=\rho^{ABC}.
\end{equation}
The main question we want answer in this section can be formulated
as follows: \emph{Are preshared correlations useful for entanglement
distribution?}

%As we will see in the following,
The answer to this question is negative for any convex entanglement
measure $E$ if Alice and Bob initially share a separable state: 
\begin{equation}
\rho_{i}=\sum_{k}p_{k}\rho_{k}^{AC}\otimes\rho_{k}^{B}.\label{eq:separable}
\end{equation}
In this case the initial entanglement is zero, and the amount of distributed
entanglement is thus given by $E^{A|BC}(\rho_{f})=E^{A|BC}(\sum_{k}p_{k}\Lambda^{C}[\rho_{k}^{AC}]\otimes\rho_{k}^{B})$.
For any convex entanglement quantifier $E$ these arguments imply
that any separable initial state $\rho_{i}$ given in Eq. (\ref{eq:separable})
can be outperformed by some pure state $\ket{\psi}^{AC}$: 
\begin{equation}
E^{A|C}(\Lambda^{C}[\proj{\psi}^{AC}])\geq E^{A|BC}(\Lambda^{C}[\rho_{i}]).
\end{equation}
This proves that preshared correlations are not useful for entanglement
distribution for any convex entanglement measure if the preshared
state is separable.

\subsection{\label{sub:subadditive}Subadditive entanglement measures}

We will now consider subadditive entanglement measures, which are
the ones that satisfy the following inequality 
\begin{equation}
E^{A_{1}A_{2}|B_{1}B_{2}}(\rho^{A_{1}B_{1}}\otimes\sigma^{A_{2}B_{2}})\leq E^{A_{1}|B_{1}}(\rho^{A_{1}B_{1}})+E^{A_{2}|B_{2}}(\sigma^{A_{2}B_{2}})\label{eq:subadditive}
\end{equation}
for any two states $\rho^{A_{1}B_{1}}$ and $\sigma^{A_{2}B_{2}}$.
Well known examples of such measures are the entanglement of formation,
the relative entropy of entanglement, and the logarithmic negativity,
which in fact is additive, i.e., it satisfies Eq. (\ref{eq:subadditive})
with equality. Moreover, we will also consider single-qubit Pauli
channels $\Lambda_{\mathrm{p}}^{C}$ which were already introduced
in Sec. \ref{sub:Pauli}. As it is proven in the following theorem,
for this type of noise preshared correlations are not useful if the
corresponding entanglement quantifier is subadditive. 
\begin{thm}
\label{thm:subadditive}Given a single-qubit Pauli channel $\Lambda_{\mathrm{p}}^{C}$
and two states $\rho_{i}=\rho^{ABC}$ and $\rho_{f}=\Lambda_{\mathrm{p}}^{C}[\rho_{i}]$,
the following inequality holds for any subadditive entanglement measure
$E$: 
\begin{equation}
E^{A|C}(\Lambda_{\mathrm{p}}^{C}[\proj{\phi^{+}}^{AC}])\geq E^{A|BC}(\rho_{f})-E^{AC|B}(\rho_{i}).
\end{equation}
\end{thm}
\begin{proof}
The proof goes along the same lines as the one of Theorem \ref{thm:Pauli}.
We denote the initial state by $\rho_{i}=\rho^{ABR}$, where Alice
is in possession of the particle $A$ and a qubit $R$, and Bob is
in possession of the remaining particle $B$. Additionally, Alice
and Bob have access to the qubits $\tilde{R}$ and $C$ of the state
$\Lambda_{\mathrm{p}}^{C}[\proj{\phi^{+}}^{\tilde{R}C}]$. By applying
the standard teleportation protocol \cite{Bennett1993} to teleport
the qubit $R$ from Alice to Bob and using the state $\Lambda_{\mathrm{p}}^{C}[\proj{\phi^{+}}^{\tilde{R}C}]$
as the resource, we see that Alice and Bob end up in the final state
$\rho_{f}=\Lambda_{\mathrm{p}}^{C}[\rho^{ABC}]$. Using the fact that
all steps in the standard teleportation protocol can be performed
by LOCC and that the amount of entanglement cannot increase in this
process, it follows that the final entanglement $E^{A|BC}(\rho_{f})$
is bounded from above by the amount of entanglement in the total initial
state: $E^{A|BC}(\rho_{f})\leq E^{AR\tilde{R}|BC}(\rho^{ABR}\otimes\Lambda_{\mathrm{p}}^{C}[\proj{\phi^{+}}^{\tilde{R}C}])$.
Finally, for a subadditive entanglement quantifier we can apply Eq.
(\ref{eq:subadditive}), which gives us $E^{A|BC}(\rho_{f})\leq E^{AR|B}(\rho^{ABR})+E^{\tilde{R}|C}(\Lambda_{\mathrm{p}}^{C}[\proj{\phi^{+}}^{\tilde{R}C}])$.
To complete the proof it is enough to notice that the state $\rho^{ABR}$
is equivalent to the initial state $\rho_{i}$. 
\end{proof}
From the theorem it also follows that sending one half of a maximally
entangled state down the noisy Pauli channel is the optimal strategy.
Let us notice that similarly as in Sec. \ref{sub:Pauli} the above
result can further be generalized to the scenario in which the exchanged
particle $C$ consists of $n$ qubits, and the channel $\Lambda_{\mathrm{p}}^{C}$
is a tensor product of $n$ (possibly different) single-qubit Pauli
channels. For any subadditive entanglement quantifier preshared correlations
do not provide any advantage also in this scenario, and the best performance
is achieved for the maximally entangled state.

Note that these arguments also cover the situation where the channel
used for entanglement distribution is noiseless. On the other hand,
if the measure of entanglement is \emph{not} subadditive, preshared
correlations can indeed be helpful even in the noiseless scenario.
This can be seen by considering the second power of the entanglement
of formation: $E=E_{f}^{2}$. Note that $E$ is a proper entanglement
quantifier, i.e., it is nonincreasing under LOCC and zero only on
separable states. If Alice and Bob have access to a noiseless single-qubit
channel and do not share any initial correlations, the optimal strategy
for Alice is to prepare locally two qubits $A$ and $C$ in the maximally
entangled state $\ket{\phi^{+}}^{AC}$, and to send the qubit $C$
to Bob. The amount of entanglement distributed in this way is given
by $E^{A|C}(\ket{\phi^{+}}^{AC})=1$. However, Alice and Bob can achieve
a better performance if they initially share the state $\ket{\psi}=\ket{\psi}^{ABC}=(\ket{000}+\ket{101}+\ket{210}+\ket{311})/2$.
In this case the amount of distributed entanglement is given by $E^{A|BC}(\ket{\psi})-E^{AC|B}(\ket{\psi})=3$.

\subsection{\label{sub:distillable-entanglement}Distillable entanglement}

The results presented so far can also be extended to the distillable
entanglement $E_{d}$ which was conjectured to be superadditive in
\cite{Shor2001}, i.e., it violates the inequality (\ref{eq:subadditive})
for some states. Based on this conjecture we will now show that preshared
correlations can be useful for the distribution of distillable entanglement.
In particular, we will consider \emph{entanglement binding channels}
$\Lambda_{\mathrm{eb}}^{C}$, i.e., channels that destroy the distillable
entanglement in any initial state $\rho^{AC}$: $E_{d}^{A|C}(\Lambda_{\mathrm{eb}}^{C}[\rho^{AC}])=0$.
This implies that this type of channels cannot be used for the distribution
of distillable entanglement if Alice and Bob do not share any correlations
initially. However, provided the superadditivity conjecture is true,
one can show that entanglement binding channels can still be used
for entanglement distribution if preshared correlations are available. 
\begin{conjecture}
There exists a state $\rho_{i}=\rho^{ABC}$ and an entanglement binding
channel $\Lambda_{\mathrm{eb}}^{C}$ for which the following inequality
holds 
\begin{equation}
E_{d}^{A|BC}(\Lambda_{\mathrm{eb}}^{C}[\rho_{i}])>E_{d}^{AC|B}(\rho_{i}).
\end{equation}

\end{conjecture}
In the following we will prove this conjecture, assuming the validity
of the superadditivity conjecture for distillable entanglement. To
this end, we consider two bound entangled states $\rho_{\mathrm{be}}^{X_{1}Y_{1}}$
and $\sigma_{\mathrm{be}}^{X_{2}Y_{2}}$ for which 
\begin{equation}
E_{d}^{X_{1}X_{2}|Y_{1}Y_{2}}(\rho_{\mathrm{be}}^{X_{1}Y_{1}}\otimes\sigma_{\mathrm{be}}^{X_{2}Y_{2}})>0,
\end{equation}
and assume that initially Alice and Bob share one of them, say $\rho_{\mathrm{be}}$.
In the next step, Alice and Bob use an entanglement binding channel
to establish the additional state $\sigma_{\mathrm{be}}$ between
them. The existence of such an entanglement binding channel is assured
by results provided in \cite{Horodecki2000a}. After this procedure
Alice and Bob share the conjectured distillable state $\rho_{\mathrm{be}}\otimes\sigma_{\mathrm{be}}$.
As a consequence, entanglement binding channels can be used for entanglement
distribution in the presence of preshared correlations under the assumption
that the distillable entanglement is superadditive.

\subsection{Distance-based entanglement measures}

\begin{table*}
\begin{centering}
\begin{tabular}{>{\raggedright}p{4cm}|>{\raggedright}p{4cm}|>{\raggedright}p{4cm}|>{\raggedright}p{4cm}|>{\raggedright}p{1cm}}
\textbf{Entanglement measure}  & \textbf{Type of noisy channel}  & \textbf{Optimal states (without preshared correlations)}  & \textbf{Advantage of preshared correlations} & \textbf{Section}\tabularnewline
\hline 
\hline 
Subadditive entanglement measures  & Single-qubit Pauli channel or any combination thereof  & Maximally entangled states  & No advantage & \ref{sub:subadditive}\tabularnewline
\hline 
Entanglement measures which are not subadditive & Noiseless channel  & Maximally entangled states  & Some of these measures show advantage even in the noiseless scenario & \ref{sub:subadditive}\tabularnewline
\hline 
Distillable entanglement  & Entanglement binding channels  & Without preshared correlations no entanglement distribution possible  & Conjectured advantage, based on the superadditivity conjecture of
distillable entanglement & \ref{sub:distillable-entanglement}\tabularnewline
\hline 
Measures which for two qubits reduce to a nondecreasing function of
concurrence  & Single-qubit noise  & Maximally entangled states  & Unknown & \ref{sub:EoF}\tabularnewline
\hline 
Logarithmic negativity  & Single-qubit amplitude damping channel  & Maximally entangled states are not always optimal  & Unknown & \ref{sub:negativity}\tabularnewline
\end{tabular}
\par\end{centering}

\caption{\label{tab:results} Overview over the entanglement quantifiers and
types of noisy channels considered in this paper. The table shows
also the optimal state for entanglement distribution without preshared
correlations for the corresponding entanglement measure and quantum
channel. As can be seen from the fourth column, in some situations
preshared correlations show an advantage for entanglement distribution.
The last column shows the section in this article, where the corresponding
result has been obtained.}
\end{table*}

In the last part of this section we consider distance-based entanglement
quantifiers $E$, as defined in Eq. (\ref{eq:entanglement}). We have
the following (without loss of generality we assume that $d_{A}\geq d_{C}$): 
\begin{thm}
\label{thm:bound}For any noisy channel $\Lambda^{C}$ there exists
a pure state $\ket{\psi}=\ket{\psi}^{AC}$ such that the following
inequality holds for any two states $\rho_{i}=\rho^{ABC}$ and $\rho_{f}=\Lambda^{C}[\rho_{i}]$:
\begin{equation}
\Delta^{C|A}(\Lambda^{C}[\proj{\psi}])\geq E^{A|BC}(\rho_{f})-E^{AC|B}(\rho_{i}).\label{thm:10}
\end{equation}
\end{thm}
\begin{proof}
From Theorem \ref{thm:noisy} it follows that the amount of distributed
entanglement is bounded above by the amount of discord between the
exchanged particle $C$ and the remaining system $AB$ in the final
state $\rho_{f}=\Lambda^{C}[\rho_{i}]$: 
\begin{equation}
\Delta^{C|AB}(\Lambda^{C}[\rho_{i}])\geq E^{A|BC}(\rho_{f})-E^{AC|B}(\rho_{i}).\label{nierownosc}
\end{equation}
Then, to obtain Eq. (\ref{thm:10}) from Eq. (\ref{nierownosc}) it
suffices to show that for any channel $\Lambda^{C}$ there exists
a pure state $\ket{\psi}=\ket{\psi}^{AC}$ such that the following
inequality 
\begin{equation}
\Delta^{C|A}(\Lambda^{C}[\proj{\psi}])\geq\Delta^{C|AB}(\Lambda^{C}[\rho_{i}])\label{eq:optimal discord}
\end{equation}
holds for any initial state $\rho_{i}=\rho^{ABC}$. For this purpose,
let us first denote by $\ket{\phi}=\ket{\phi}^{ABCR}$ the purification
of $\rho_{i}$, i.e., $\rho_{i}=\mathrm{Tr}_{R}[\proj{\phi}]$. Then,
we recall that all distance-based quantifiers of discord $\Delta^{X|Y}$
do not increase under quantum operations on the subsystem $Y$, if
the corresponding distance does not increase under quantum operations
\cite{Streltsov2011a}. This implies that the inequality $\Delta^{C|ABR}(\Lambda^{C}[\proj{\phi}])\geq\Delta^{C|AB}(\Lambda^{C}[\rho_{i}])$
is satisfied. The proof of Eq. (\ref{eq:optimal discord}) is complete
by recalling that $d_{A}\geq d_{C}$, and thus there must exist a
pure state $\ket{\psi}=\ket{\psi}^{AC}$ such that $\Delta^{C|A}(\Lambda^{C}[\proj{\psi}])\geq\Delta^{C|ABR}(\Lambda^{C}[\proj{\phi}])$
is true for any state $\ket{\phi}=\ket{\phi}^{ABCR}$. 
\end{proof}
Let us now apply the above result to the single-qubit phase damping
channel $\Lambda_{\mathrm{pd}}^{C}$ which is a special case of a
Pauli channel and is defined as follows: 
\begin{equation}
\Lambda_{\mathrm{pd}}^{C}[\rho_{i}]=(1-p)\cdot\rho_{i}+p\cdot\sigma_{z}^{C}\rho_{i}\sigma_{z}^{C}
\end{equation}
with an initial state $\rho_{i}=\rho^{ABC}$ and the damping parameter
$p$ ranging from $0$ to $1/2$. While $p=0$ corresponds to the
noiseless scenario, full phase damping is achieved for $p=1/2$. Using
Theorem \ref{thm:subadditive}, it follows that for this type of noise
maximally entangled states are optimal for entanglement distribution
if the quantifier of entanglement is subadditive. In particular, this
is true for the relative entropy of entanglement $E_{R}$ defined
in Eq. (\ref{eq:REE}). As we will see in the following, for this
entanglement measure the bound provided in Theorem \ref{thm:bound}
turns out to be tight for any single-qubit phase damping channel:
\begin{align}
\Delta_{R}^{C|A}(\Lambda_{\mathrm{pd}}^{C}[\proj{\phi^{+}}^{AC}]) & =E_{R}^{A|C}(\Lambda_{\mathrm{pd}}^{C}[\proj{\phi^{+}}^{AC}])\label{eq:phase damping}\\
 & \geq E_{R}^{A|BC}(\rho_{f})-E_{R}^{AC|B}(\rho_{i}).\nonumber 
\end{align}
Here, $\Delta_{R}$ is the relative entropy of discord defined in
Eq. (\ref{eq:RED}), $\rho_{i}=\rho^{ABC}$ is an arbitrary initial
state with $d_{A}\geq d_{C}=2$, and $\rho_{f}=\Lambda_{\mathrm{pd}}^{C}[\rho_{i}]$
is the final state after the application of the noisy channel.

To prove Eq. (\ref{eq:phase damping}) let us notice that the following
chain of inequalities holds:
\begin{align}
S(\rho^{XY}||\sum_{i}\Pi_{i}^{X}\rho^{XY}\Pi_{i}^{X}) & \geq\Delta_{R}^{X|Y}(\rho^{XY})\geq E_{R}^{X|Y}(\rho^{XY})\nonumber \\
 & \geq S(\rho^{X})-S(\rho^{XY}),\label{eq:bound-1}
\end{align}
where $\{\Pi_{i}^{X}\}$ is a local von Neumann measurement on the
particle $X$, and the last inequality was proven in \cite{Plenio2000}.
If we now choose the projectors $\Pi_{i}^{C}=\proj{i}^{C}$, it can
be verified that for the state $\sigma^{AC}=\Lambda_{\mathrm{pd}}^{C}[\proj{\phi^{+}}^{AC}]$
the upper and the lower bound in Eq. (\ref{eq:bound-1}) coincide:
$S(\sigma^{AC}||\sum_{i}\Pi_{i}^{C}\sigma^{AC}\Pi_{i}^{C})=S(\sigma^{C})-S(\sigma^{AC})$.
Together with Theorem \ref{thm:subadditive} this completes the proof
of Eq. (\ref{eq:phase damping}). In particular, this also shows that
the bound provided in Theorem \ref{thm:bound} is tight for single-qubit
phase damping channels, since for this type of noise the amount of
distributed entanglement is bounded above by $\Delta_{R}^{C|A}(\Lambda_{\mathrm{pd}}^{C}[\proj{\phi^{+}}^{AC}])$,
and this bound is also reachable according to Eq. (\ref{eq:phase damping}).

\section{Conclusions and outlook}

A concise summary of our results is presented in Table \ref{tab:results},
where we list several entanglement quantifiers and types of noisy
channels considered in this work, show the corresponding optimal state
and discuss the advantage of preshared correlations. In two of the
cases it remains unclear if preshared correlations provide an advantage
for entanglement distribution. We leave this question open for future
research.

The results presented in this work can be regarded as a major step
towards a unified approach to entanglement distribution. In particular,
it can be seen from the first row in Table \ref{tab:results} that
preshared correlations do not provide advantage for any subadditive
entanglement quantifier, if entanglement is distributed via a combination
of single-qubit Pauli channels. In this context, it is tempting to
assume that these results extend to arbitrary noisy channels, and
thus preshared correlations do not provide advantage for any subadditive
entanglement measure and any type of noise\emph{.} Sending one half
of a pure entangled state down a noisy channel would then be the optimal
strategy for any subadditive entanglement measure. While we cannot
prove this conjecture in full generality at this point, our results
strongly support this statement. In particular, advantage of preshared
correlations was only found for measures which are not subadditive,
and for distillable entanglement which is conjectured to be superadditive.

Regarding entanglement distribution with separable states, our results
show that this strategy is not reasonable for any subadditive entanglement
measure, if a combination of single-qubit Pauli channels is used for
the process. On the other hand, this result does not rule out the
superiority of separable states for other types of noise. In this
direction we have found, supplementing the results of Ref. \cite{Pal},
that states with arbitrarily little entanglement can outperform maximally
entangled states for amplitude damping noise, if entanglement is quantified
via the logarithmic negativity. These counterintuitive results also
imply that a closer investigation of entanglement distribution with
separable states is necessary, since -- contrary to recent claims
made e.g. in \cite{Kay2012,Fedrizzi2013} -- maximally entangled states
are not necessarily the best resource to benchmark this process.

The results of this paper can also be seen as the first step toward
similar considerations in quantum many body systems. Note that over
the last decade entanglement has proven to be extremely useful to
characterize properties of many body systems and the nature of quantum
phase transitions \cite{Sachdev1999}. For instance, in the ground
states and low energy states of quantum spin models the following
properties hold (see \cite{Augusiak2012,Lewenstein2012} for a review): 
\begin{itemize}
\item The two body reduced density matrix typically exhibits entanglement
for short separations of the spins only, even at criticality; still,
entanglement measures show signatures of quantum phase transitions
\cite{Osterloh2002,Osborne2002}. 
\item One can concentrate entanglement between the chosen two spins by optimized
measurements on the rest of the system, obtaining in this way the
so-called \textit{localizable entanglement} \cite{Verstraete2004};
the corresponding entanglement length diverges when the standard correlation
length diverges, i.e., at standard quantum phase transitions. 
\item For non-critical systems, ground states and low energy states exhibit
area laws: the von Neumann or Rényi entropy of the reduced density
matrix of a block of size $I$ scales as the size of the boundary
of the block, $\partial I$; at criticality logarithmic divergence
occurs frequently \cite{Vidal2003} (see also \cite{Calabrese2009,Eisert2010}
for a review). 
\item Ground states and low energy states can be efficiently described by
matrix product states, or more generally tensor network states (cf.
\cite{Verstraete2006}). 
\item Topological order for gapped systems in 1D and 2D exhibits itself
frequently in the properties of the \textit{entanglement spectrum},
i.e., the spectrum of the logarithm of the reduced density matrix
of a block $I$ \cite{Li2008}, and in the appearance of the \textit{topological
entropy}, i.e., negative constant correction to the area laws in 2D
\cite{Kitaev2006,Levin2006}. 
\end{itemize}
All the above results indicate the importance of few body entanglement
in the low energy physics of many body systems (cf. \cite{Guhne2005,Hofmann2014,Stasinska2014}).
It is to be expected that few body entanglement will also play a role
in characterizing out-of-equilibrium dynamics of quantum many body
systems \cite{Coser}.

Note that the scheme of entanglement distribution discussed in this
paper -- at least in the noiseless case -- can be considered in the
context of \textit{the real transfer of the particle $C$} to Bob,
or as the \textit{change of partition from $AC|B$ to $A|BC$}. In
this sense, one can view our results as a characterization of entanglement
in three-body reduced density matrix in a many-body system. Generalizations
including noisy transfer are possible, for instance by coupling $C$
locally to a reservoir or an ancillary particle. It would also be
interesting to consider the entanglement distribution scheme with
many recipients (Bobs). Finally, asking analogous questions for Bell
nonlocality or steering seems to be a fascinating perspective that
would also lead to a better understanding of these phenomena. 
\begin{acknowledgments}
We thank Dagmar Bruß, Hermann Kampermann, and Marco Piani for discussion,
and an anonymous referee for helpful suggestions. This work was supported
by the Spanish Ministry projects FOQUS and DIQIP CHIST-ERA. We acknowledge
also EU IP SIQS, EU Grant QUIC, ERC AdG OSYRIS, John Templeton Foundation,
Alexander von Humboldt-Foundation for the Feodor Lynen scholarship
for A. S, and Spanish Ministry for the Juan de la Cierva scholarship
for R. A. 
\end{acknowledgments}
\bibliographystyle{apsrev4-1} 
\end{document}